\documentclass[conference]{IEEEtran}

\IEEEoverridecommandlockouts

\usepackage{cite}
\usepackage{amsmath,amssymb,amsfonts,amsthm,bm}
\usepackage{algorithmic}
\usepackage{graphicx}
\usepackage{subcaption}
\usepackage{textcomp}
\usepackage{xcolor}
\usepackage{booktabs}
\usepackage{mathtools}
\usepackage{hyperref}
\usepackage{mathrsfs}

\hypersetup{
  colorlinks=true,
  linkcolor=black,
  citecolor=black,
  urlcolor=blue
}

\newtheorem{theorem}{Theorem}
\newtheorem{proposition}{Proposition}
\newtheorem{lemma}{Lemma}

\newtheorem{remark}{Remark}
\newtheorem{definition}{Definition}

\newcommand{\E}{\mathbb{E}}

\newcommand{\artanh}{\operatorname{artanh}}

\begin{document}

\title{Sensing with Random Signals: The Role of Time Sharing}

\author{
  \IEEEauthorblockN{Yi Geng, Wenyi Zhang}
  \IEEEauthorblockA{
    Department of Electronic Engineering and Information Science,\\
    University of Science and Technology of China,\\
    Hefei, Anhui 230027, China\\
    Email: bxzymy@mail.ustc.edu.cn, wenyizha@ustc.edu.cn
  }
}

\maketitle

\begin{abstract}
In monostatic, decision-aided, or known-waveform integrated sensing and communications (ISAC) formulations, the sensing receiver is often modeled as knowing the transmitted waveform. This assumption is not suitable for passive, bistatic, or distributed settings where the sensing receiver knows the signaling rule but not the transmitted symbols. We study such a symbol-unaware ISAC model, where sensing is measured by the unconditioned mutual information $I(S;V)$ rather than the symbol-aware quantity $I(S;V|X)$.
For discrete-input memoryless channels, we characterize the capacity--sensing region through an auxiliary time-sharing variable, showing that the optimal upper boundary is the upper concave envelope of the single-mode frontier. Thus, explicit time sharing is unnecessary when the single-mode frontier is already concave, but strictly beneficial when its upper concave envelope strictly dominates the frontier. For Rayleigh-fading BPSK, we further show that the curvature of the single-mode boundary is determined by the stochastic ordering of the communication- and sensing-side effective SNR distributions. Communication-side dominance yields a concave single-mode frontier and no time-sharing gain, sensing-side dominance yields a convex single-mode frontier and a strict time-sharing gain, and equality yields a linear boundary. The result extends to SIMO-BPSK through the ordering of post-combining SNR distributions. These findings explain when symbol-unaware ISAC optimally moves from data-symbol transmission to pilot-like sensing modes.
\end{abstract}


\section{Introduction}
ISAC enables communication and environmental sensing over shared spectrum, hardware, and waveform resources \cite{Liu2022JSAC}. Many information-theoretic and waveform-design studies assume that the sensing receiver knows the transmitted waveform or symbol realization, as in monostatic, feedback-assisted, or decision-aided architectures. Under this known-waveform premise, sensing is naturally evaluated through symbol-aware metrics and optimized through waveform, input, or beamforming design \cite{Ahmadipour2024ISAC,Xiong2023GaussianISAC,Liu2023DeterministicRandom,Li2024TVT,Wei2026JSAC,Chen2026TWC}. The premise is less suitable for passive, bistatic, or distributed ISAC, where the sensing receiver may know the signaling rule but not the realized symbols. The communication signal then remains part of the sensing uncertainty, so the relevant metric is \(I(S;V)\), not \(I(S;V|X)\). Related works on bistatic ISAC, random-signal sensing, relay-assisted sensing, and mutual-information-based sensing have shown that this regime can lead to different tradeoffs \cite{Jiao2025BistaticISAC,YLiu2026TIT,Xie2024RandomSMI,SLu2026TN,Fliu2026JSAC,JChen2026TCOM}.

This paper studies the geometry of the symbol-unaware communication--sensing frontier. Existing studies mainly evaluate fixed signaling modes, use estimation-, distortion-, or log-loss-oriented criteria, or treat time sharing as an auxiliary benchmark rather than as part of the frontier characterization itself \cite{Ahmadipour2024ISAC,Xiong2023GaussianISAC,Liu2023DeterministicRandom,Jiao2025BistaticISAC,YLiu2026TIT,JChen2026TCOM}. It remains unclear whether a single input distribution is sufficient under the metric \(I(S;V)\), or whether explicit convexification through time sharing can be fundamentally necessary.

We answer this question information-theoretically. For finite-alphabet memoryless symbol-unaware ISAC channels, we characterize the exact single-letter capacity--sensing region and show that its upper boundary is the upper concave envelope of the single-mode frontier. We then analyze Rayleigh-fading BPSK and SIMO-BPSK models, where the frontier geometry is governed by the ordering of communication- and sensing-side effective SNR distributions. Communication-side dominance yields a concave frontier and no time-sharing gain; sensing-side dominance yields a convex frontier with strict time-sharing gain; equality yields a linear frontier. Thus, under \(I(S;V)\)-based symbol-unaware sensing, time sharing can be an intrinsic component of the optimal ISAC boundary rather than a secondary design option.

\section{System Model}

Consider a symbol-unaware sensing ISAC system over a state-dependent memoryless channel as shown in Fig.1. A message $W \in [1:2^{nR}]$ is transmitted over $n$ channel uses. In addition to reliable communication, the transmitted signal probes an i.i.d.\ state sequence $S^n = (S_1,\dots,S_n)$ through sensing observations $V^n = (V_1,\dots,V_n)$.

\begin{figure}[htbp]  
    \centering
    \includegraphics[width=0.4\textwidth]{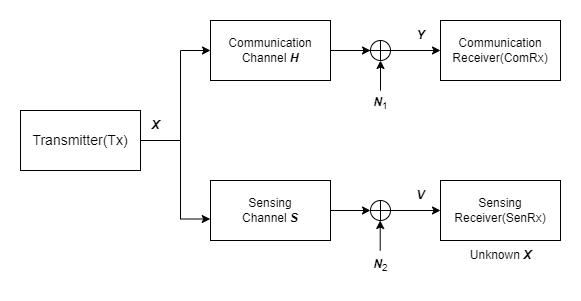}
    \caption{symbol-unaware sensing ISAC system}
    \label{fig:blind-sensing ISAC system}
\end{figure}

Unlike distortion-based formulations, we do not introduce an explicit estimator. Instead, sensing performance is quantified by the mutual information between the state and sensing observation without conditioning on the transmitted symbols, capturing the symbol-unaware sensing nature.

The physical channel is a state-dependent memoryless channel
specified by $(\mathcal X,\mathcal S,\mathcal Y,\mathcal V,P_S,W),$ where \(S_i\sim P_S\) is i.i.d. and independent of the message, and $W(y,v|x,s)=P_{Y,V|X,S}(y,v|x,s)$ is an arbitrary joint transition law from \((X,S)\) to \((Y,V)\).
Thus, for a length-\(n\) input sequence,
\begin{equation}
P(s^n,y^n,v^n|x^n)
=
\prod_{i=1}^n
P_S(s_i)W(y_i,v_i|x_i,s_i).
\label{eq:general_joint_channel}
\end{equation}

The induced communication and sensing marginal laws are
\begin{align}
W_c(y|x)
&\triangleq
\sum_{s,v} P_S(s) W(y,v|x,s), \\
W_s(v|x,s)
&\triangleq
\sum_y W(y,v|x,s),
\end{align}
or equivalently $P_{S,V|X}(s,v|x)=P_S(s)W_s(v|x,s).$
The product model $P_S(s)W_c(y|x)W_s(v|x,s)$ is only a special coupling of these two marginals and is not assumed in the general model.

\begin{definition}
A $(2^{nR},n)$ code consists of: 
(i) a message set $\mathcal{W}=[1:2^{nR}]$, 
(ii) an encoder $f_n:\mathcal{W}\to\mathcal{X}^n$, 
and (iii) a decoder $\phi_n:\mathcal{Y}^n\to\mathcal{W}$.
\end{definition}

For an input distribution $P_X$, define the communication and sensing metrics as $R(P_X)=I(X;Y)$ and $D(P_X)=I(S;V)$.

For a given code, the induced marginal at time $i$ is 
$P_{X_i}^{(n)}(x)=\Pr[X_i=x]$, and the sensing performance is defined as
\begin{equation}
D_n(f_n)
=
\frac{1}{n}\sum_{i=1}^n D\!\big(P_{X_i}^{(n)}\big).
\label{eq:Dn_def}
\end{equation}

\begin{definition}
$P_e^{(n)}$ denotes the average probability of error. A pair $(R,D)$ is achievable if there exists a sequence of $(2^{nR},n)$ codes such that $P_e^{(n)}\to 0$ and $\liminf_{n\to\infty}D_n(f_n)\ge D$.
\end{definition}

\begin{remark}
The use of $D_n(f_n)=\frac{1}{n}\sum_{i=1}^n I(S_i;V_i)$ follows directly from the memoryless sensing model. Since each state $S_i$ is observed through its corresponding sensing output $V_i$, a symbol-by-symbol sensing metric is natural, particularly for practical receivers that process sensing observations per channel use or per snapshot. This averaged single-letter metric preserves the symbol-unaware sensing effect of the unknown transmitted symbol and does not require starting from the block mutual information $I(S^n;V^n)$.
\end{remark}

\section{Single-Letter Communication--Sensing Region under Averaged Symbol-unaware SMI}

We now characterize the communication--sensing region of the symbol-unaware ISAC model in single-letter form. 
The auxiliary random variable $Q$ represents a time-sharing schedule across channel uses. 
Since the sensing metric is defined symbol by symbol, $Q$ should be understood as a per-coordinate scheduling variable rather than as a block-level sensing operation.

\begin{proposition}[Marginal invariance]
\label{prop:marginal_invariance}
Fix \(P_S\). Consider two memoryless joint transition laws
\(W\) and \(\widetilde W\) from \((X,S)\) to \((Y,V)\). Suppose that they
induce the same communication marginal and the same sensing marginal, i.e.,
for all \(x,y\),
\begin{equation}
\sum_{s,v}P_S(s)W(y,v|x,s)
=
\sum_{s,v}P_S(s)\widetilde W(y,v|x,s),
\end{equation}
and for all \(x,s,v\),
\begin{equation}
\sum_y W(y,v|x,s)
=
\sum_y \widetilde W(y,v|x,s).
\end{equation}
Then the achievable communication--sensing region under the averaged
symbol-unaware sensing metric in \eqref{eq:Dn_def} is identical for
\(W\) and \(\widetilde W\).
\end{proposition}
\begin{proof}
The proof is deferred to the extended arXiv version due to space limitations.
\end{proof}

\begin{theorem}
\label{thm:main_region}
For the memoryless symbol-unaware sensing ISAC model with arbitrary joint transition
law \(W(y,v|x,s)\), the achievable communication--sensing region, and
conversely the closure of all achievable pairs, is
\begin{equation}
\mathcal C =
\bigcup_{\substack{P_QP_{X|Q}:\\ |\mathcal Q|\le 2}}
\left\{(R,D):
\begin{aligned}
&0\le R\le I(X;Y|Q),\\
&0\le D\le I(S;V|Q)
\end{aligned}
\right\},
\label{eq:region_main}
\end{equation}
where \(I(X;Y|Q)\) is computed with the communication marginal \(W_c\), and
\(I(S;V|Q)\) is computed with the sensing marginal \(P_S W_s\). Consequently,
\(\mathcal C\) depends on the physical joint channel \(W(y,v|x,s)\) only
through the two marginals \(W_c(y|x)\) and \(W_s(v|x,s)\).
Equivalently, the dominant tradeoff boundary of \(\mathcal C\) is the upper
concave envelope of the single-mode frontier generated by input distributions
\(P_X\), and the full region is the downward closure of this boundary. Hence,
when the endpoint values are nonzero, the complete outer boundary also includes
the horizontal and vertical endpoint extensions induced by this downward closure.
\end{theorem}

\begin{proof}
For any \(P_QP_{X|Q}\), the communication quantity \(I(X;Y|Q)\) depends only on
the marginal channel \(P_{Y|X}=W_c\). Similarly, the symbol-unaware sensing quantity
\(I(S;V|Q)\) depends only on the marginal law \(P_{S,V|X}=P_SW_v\). Hence the
set of single-mode operating points
\begin{equation}
\mathcal A=
\left\{(I(X;Y),I(S;V)):P_X\in\mathcal P(\mathcal X)\right\} \subset \mathbb R^2 
\end{equation}
is the same under \(W\) and \(\widetilde W\). Since the achievable region is the
downward closure of the convex hull of \(\mathcal A\), the two regions are
identical.
\end{proof}

Geometrically, the upper boundary is given by the upper concave envelope, while the full region also contains the horizontal and vertical boundary segments induced by the downward closure.

\begin{remark}
The cardinality bound $|\mathcal Q|\le 2$ follows from the Fenchel--Eggleston strengthening of Carathéodory's theorem. Indeed, the set $\mathcal A$ is compact and connected, since $\mathcal P(\mathcal X)$ is compact and connected and mutual information is continuous in $P_X$. Hence every point in $\operatorname{conv}(\mathcal A)$ can be represented as a convex combination of at most two points in $\mathcal A$, which corresponds to a binary time-sharing variable $Q$.
\end{remark}

\begin{remark}[Relation to~\cite{WZhang2011} and the role of $Q$]
The proof follows the constrained channel-coding argument of~\cite{WZhang2011}, but the role of the sensing constraint is different. In~\cite{WZhang2011}, after optimizing the symbolwise estimator, the sensing constraint becomes an expected per-input cost, which is affine in $P_X$. Hence any time sharing can be absorbed into the averaged input distribution without decreasing the communication rate, since $I(X;Y)$ is concave in $P_X$. An explicit auxiliary variable $Q$ is therefore unnecessary there.

In the present blind-sensing model, the sensing metric is instead $D(P_X)=I(S;V),$ which is generally convex, not affine, in $P_X$. Thus, for
\(
\bar P_X=\sum_q P_Q(q)P_{X|Q=q},
\)
one only has
\begin{equation}
D(\bar P_X)
\le
\sum_q P_Q(q)D(P_{X|Q=q})
=
I(S;V|Q).
\end{equation}
Absorbing $Q$ into a single input distribution may therefore reduce the blind-sensing value. This is why \(Q\) is generally needed in \eqref{eq:region_main}, and why the optimal upper boundary is the upper concave envelope of the single-mode frontier.
\end{remark}

\begin{proof}
We give the main steps, since the reliability part follows the standard constant-composition coding argument for constrained channel coding.

\subsection*{Achievability}

Fix $P_QP_{X|Q}$ with finite $\mathcal Q$. Choose a deterministic time-sharing sequence $q^n=(q_1,\ldots,q_n)$ whose empirical distribution tends to $P_Q$, and define $\mathcal I_q=\{i:q_i=q\}$. Choose a conditional type $P_{X|Q,n}$ compatible with $q^n$ such that
\begin{equation}
P_{X|Q,n}(\cdot|q)\to P_{X|Q}(\cdot|q),
\qquad \forall q\in\mathcal Q.
\end{equation}
Generate each length-$n$ codeword uniformly from the conditional type class
$T(P_{X|Q,n}|q^n)$, so that for every $x\in\mathcal X$ and $q\in\mathcal Q$,
every codeword $x^n(m)$ satisfies
\begin{equation}
\frac{1}{|\mathcal I_q|}
\sum_{i\in\mathcal I_q}
\mathbf 1\{x_i(m)=x\}
=
P_{X|Q,n}(x|q).
\end{equation}

It remains to verify the sensing metric. For a fixed input distribution $P_X$, define $D(P_X)\triangleq I(S;V)$ under the law induced by $P_X(x)P_S(s)W_s(v|x,s)$. For fixed $P_S$ and $W_s$, $P_{V|S}$ is affine in $P_X$; since $I(S;V)$ is convex in $P_{V|S}$ for fixed $P_S$, $D(P_X)$ is convex in $P_X$.

Let $P_{X_i}^{(n)}$ denote the input marginal at coordinate $i$ induced by the uniform message over the generated codebook. The conditional-type constraint gives, for every $q\in\mathcal Q$ and $x\in\mathcal X$,
\begin{equation}
\frac{1}{|\mathcal I_q|}
\sum_{i\in\mathcal I_q}
P_{X_i}^{(n)}(x)
=
P_{X|Q,n}(x|q).
\end{equation}
Therefore, by Jensen's inequality,
\begin{equation}
\begin{aligned}
D_n(f_n)
&=
\frac1n\sum_{i=1}^n D(P_{X_i}^{(n)}) \\
&=
\sum_{q\in\mathcal Q}
\frac{|\mathcal I_q|}{n}
\frac{1}{|\mathcal I_q|}
\sum_{i\in\mathcal I_q}
D(P_{X_i}^{(n)}) \\
&\ge
\sum_{q\in\mathcal Q}
\frac{|\mathcal I_q|}{n}
D\!\left(
\frac{1}{|\mathcal I_q|}
\sum_{i\in\mathcal I_q}P_{X_i}^{(n)}
\right) \\
&=
\sum_{q\in\mathcal Q}
\frac{|\mathcal I_q|}{n}
D\!\left(P_{X|Q,n}(\cdot|q)\right).
\end{aligned}
\end{equation}
Letting $n\to\infty$ and using the continuity of mutual information in the input distribution gives
\begin{equation}
\liminf_{n\to\infty}D_n(f_n)
\ge
\sum_{q\in\mathcal Q}P_Q(q)D(P_{X|Q=q})
=
I(S;V|Q).
\end{equation}
Thus every pair satisfying
$R<I(X;Y|Q)$ and $D<I(S;V|Q)$ is achievable; the non-strict inequalities in \eqref{eq:region_main} follow by closure.

\subsection*{Converse}

Consider any achievable sequence of codes. Let $M$ be uniformly distributed over the message set and let $X^n=f_n(M)$. By Fano's inequality, the data-processing inequality, and the memoryless
property of the channel,
\begin{equation}
nR
\le
\sum_{i=1}^n I(X_i;Y_i)+n\epsilon_n,
\qquad \epsilon_n\to0.
\end{equation}
where $\epsilon_n\to0$. Hence
\begin{equation}
R
\le
\frac1n\sum_{i=1}^n I(X_i;Y_i)+\epsilon_n.
\end{equation}

Let $T$ be uniformly distributed over $\{1,\ldots,n\}$ and independent of all other variables, and define $(X,S,Y,V)=(X_T,S_T,Y_T,V_T).$ Then
\begin{equation}
\frac1n\sum_{i=1}^n I(X_i;Y_i)
=
I(X;Y|T),
\end{equation}
and, by the definition of the symbol-by-symbol symbol-unaware sensing metric,
\begin{equation}
D_n(f_n)
=
\frac1n\sum_{i=1}^n I(S_i;V_i)
=
I(S;V|T).
\end{equation}
The index $T$ only single-letterizes the $n$ coordinates of an arbitrary code; the final auxiliary variable is obtained by replacing the resulting averaged point by an equivalent cardinality-reduced time-sharing representation. Indeed, the pair $\bigl(I(X;Y|T),I(S;V|T)\bigr)$ belongs to $\operatorname{conv}(\mathcal A)$, where
\begin{equation}
\mathcal A
=
\left\{
\bigl(I(X;Y),I(S;V)\bigr):
P_X\in\mathcal P(\mathcal X)
\right\}.
\end{equation}

Since $\mathcal P(\mathcal X)$ is compact and connected, and mutual information is continuous in $P_X$, the set $\mathcal A$ is compact and connected. Hence any subsequential limit of the averaged pairs remains in $\operatorname{conv}(\mathcal A)$. By the Fenchel--Eggleston theorem, this limiting point can be represented by a time-sharing variable $Q$ with $|\mathcal Q|\le2$. Consequently,
$R\le I(X;Y|Q), \quad D\le I(S;V|Q).$
This proves the converse and completes the proof.
\end{proof}

\section{Geometry of the Single-Mode Boundary under BPSK and SIMO-BPSK}
In this section, we specialize the single-letter time-sharing structure from Section III to Rayleigh-fading BPSK and SIMO-BPSK models. Our goal is to determine whether the single-mode frontier \(C_0(D)\), generated without explicit time sharing, is already concave or whether its upper concave envelope strictly improves the achievable boundary.

\subsection{Scalar BPSK Frontier and Curvature Criterion}

Consider the scalar Rayleigh-fading BPSK model with
$X=\sqrt P B$, $B\in\{+1,-1\}$, and $\Pr(B=+1)=p$.
By the sign symmetry of BPSK, it suffices to take $p\in[1/2,1]$.
The communication and sensing observations are
\begin{equation}
Y=HX+N_1,
\label{eq:com_channel}
\end{equation}
and
\begin{equation}
V=SX+N_2,
\end{equation}
where $H\sim\mathcal{CN}(0,\Omega_H)$,
$S\sim\mathcal{CN}(0,\Omega_S)$,
$N_1\sim\mathcal{CN}(0,\sigma_1^2)$, and
$N_2\sim\mathcal{CN}(0,\sigma_2^2)$. The communication receiver knows
$H$, whereas the sensing receiver does not know the transmitted symbol
realization.

Define the scalar BPSK-AWGN mutual-information kernel
\begin{equation}
J(p,\gamma)\triangleq I(B;\sqrt{\gamma}B+Z),
\end{equation}
where $Z\sim\mathcal N(0,1)$. For each fixed $\gamma$, $J(p,\gamma)$ is
symmetric around $p=1/2$ and nonincreasing in $p$ on $[1/2,1]$.

Conditioned on $H=h$, the communication channel is a scalar BPSK-AWGN
channel with effective SNR $\frac{\gamma=2P|h|^2}{\sigma_1^2}$. Hence
\begin{equation}
R(p)\triangleq I(X;Y|H)
=
\int_0^\infty J(p,\gamma)w_c(\gamma)d\gamma,
\label{eq:R_bpsk}
\end{equation}
where $w_c(\gamma)=\frac{1}{\bar\gamma_c}e^{-\gamma/\bar\gamma_c}$, and $ \bar\gamma_c=\frac{2P\Omega_H}{\sigma_1^2}.$

For sensing, the chain rule gives
\begin{equation}
I(S;V)=I(S;V|X)+I(X;V)-I(X;V|S).
\end{equation}
Since $S$ is zero-mean circularly symmetric complex Gaussian,
$V|X=+\sqrt P$ and $V|X=-\sqrt P$ have the same distribution, and thus
$I(X;V)=0$. Therefore
\begin{equation}
G(p)\triangleq I(S;V)
=
\log_2\left(1+\frac{P\Omega_S}{\sigma_2^2}\right)
-
\int_0^\infty J(p,\gamma)w_s(\gamma)d\gamma,
\label{eq:G_bpsk}
\end{equation}
where $w_s(\gamma)=\frac{1}{\bar\gamma_s}e^{-\gamma/\bar\gamma_s},$ and $\bar\gamma_s=\frac{2P\Omega_S}{\sigma_2^2}.$.

Since $J(p,\gamma)$ is nonincreasing in $p$, $R(p)$ decreases and $G(p)$
increases with $p$. The nontrivial single-mode frontier is therefore
\begin{equation}
C_0(D)=R(G^{-1}(D)),
\qquad
D\in[D_{\min},D_{\max}],
\label{eq:C0_def_bpsk}
\end{equation}
where $D_{\min}=G(1/2)$ and $D_{\max}=G(1)$. With time sharing, the
corresponding boundary is the upper concave envelope of $C_0(D)$, and the
achievable region is its downward closure. Since \eqref{eq:C0_def_bpsk}
only covers $D\in[D_{\min},D_{\max}]$, the complete boundary also includes
the endpoint extensions
$R=R(1/2)$ for $0\le D\le D_{\min}$ and
$D=D_{\max}$ for $0\le R\le R(1)$.

The endpoint $p=1/2$ corresponds to equiprobable data symbols, whereas
$p=1$ corresponds to a deterministic pilot-like mode. Thus the curvature of
$C_0(D)$ determines whether continuous input biasing is sufficient or
explicit time sharing is needed.

Define the marginal communication cost
\begin{equation}
\eta(p)\triangleq -\frac{R'(p)}{G'(p)}.
\label{eq:eta_def}
\end{equation}
Since
\begin{equation}
C_0'(D)=\frac{R'(p)}{G'(p)}=-\eta(p),
\end{equation}
we have
\begin{equation}
C_0''(D)=-\frac{\eta'(p)}{G'(p)}.
\label{eq:C0_second_eta}
\end{equation}
Thus, on the interior interval where $G'(p)>0$, the sign of the frontier
curvature is opposite to the sign of $\eta'(p)$.

Here and in the sequel, subscripts of $J$ denote partial derivatives with
respect to the BPSK input-bias parameter $p$; that is,
$J_p(p,\gamma)=\frac{\partial J(p,\gamma)}{\partial p}$ and
$J_{pp}(p,\gamma)=\frac{\partial^2 J(p,\gamma)}{\partial p^2}$.

Let
\begin{equation}
\kappa_p(\gamma)
\triangleq
\frac{J_{pp}(p,\gamma)}{J_p(p,\gamma)}
=
\frac{\partial}{\partial p}\log[-J_p(p,\gamma)].
\label{eq:kappa_def}
\end{equation}
For $a\in\{c,s\}$, define the tilted SNR density
\begin{equation}
\mu_{a,p}(\gamma)
\triangleq
\frac{-J_p(p,\gamma)w_a(\gamma)}
{\int_0^\infty -J_p(p,t)w_a(t)dt}.
\label{eq:mu_a_def}
\end{equation}
Then
\begin{equation}
\eta'(p)
=
\eta(p)
\left(
\mathbb E_{\mu_{c,p}}[\kappa_p(\Gamma)]
-
\mathbb E_{\mu_{s,p}}[\kappa_p(\Gamma)]
\right).
\label{eq:eta_prime_bpsk}
\end{equation}
Hence the curvature problem reduces to comparing two tilted averages of the
same function $\kappa_p(\gamma)$. Equivalently, if
\begin{equation}
\Delta_\kappa(p)
\triangleq
\mathbb E_{\mu_{c,p}}[\kappa_p(\Gamma)]
-
\mathbb E_{\mu_{s,p}}[\kappa_p(\Gamma)],
\end{equation}
then
\begin{equation}
C_0''(G(p))
=
-\frac{\eta(p)}{G'(p)}\Delta_\kappa(p).
\end{equation}
Thus the sign of $\Delta_\kappa(p)$ determines the local concavity or
convexity of the single-mode frontier.

The following two lemmas turn this comparison into an ordering condition.

\begin{lemma}
\label{lem:kappa_monotone}
For every fixed $p\in(1/2,1)$, the function
$\gamma\mapsto\kappa_p(\gamma)$ is strictly increasing on $(0,\infty)$.
\end{lemma}

\begin{proof}
The proof is standard and is omitted; see the extended arXiv version for details.
\end{proof}

\begin{lemma}
\label{lem:mlr_ordering}
Let $q_1$ and $q_2$ be probability densities on an interval
$\mathcal X\subseteq\mathbb R$. If $q_1/q_2$ is nondecreasing, then for
every measurable increasing function $\phi$ with finite expectations,
\begin{equation}
\mathbb E_{q_1}[\phi(X)]\ge \mathbb E_{q_2}[\phi(X)].
\end{equation}
If both $q_1/q_2$ and $\phi$ are strictly increasing, the inequality is
strict.
\end{lemma}

\begin{proof}
The proof is deferred to the extended arXiv version due to space limitations.
\end{proof}

By Lemmas~\ref{lem:kappa_monotone} and~\ref{lem:mlr_ordering}, the sign of
$\eta'(p)$ is determined by the monotone likelihood-ratio ordering between
$\mu_{c,p}$ and $\mu_{s,p}$. In the scalar Rayleigh case, this ordering
reduces to comparing $\bar\gamma_c$ and $\bar\gamma_s$, as shown next.

\subsection{Scalar Rayleigh-BPSK Classification}

We now specialize the curvature criterion to the scalar Rayleigh-fading BPSK model. Recall that $w_a(\gamma)=\frac{1}{\bar\gamma_a}e^{-\gamma/\bar\gamma_a}$ for $a\in\{c,s\}$, where $\bar\gamma_c=\frac{2P\Omega_H}{\sigma_1^2}$ and $\bar\gamma_s=\frac{2P\Omega_S}{\sigma_2^2}$.
For the tilted SNR densities defined in \eqref{eq:mu_a_def}, their likelihood
ratio is
\begin{equation}
\frac{\mu_{c,p}(\gamma)}{\mu_{s,p}(\gamma)}
=
C_p\frac{w_c(\gamma)}{w_s(\gamma)}
=
C_p\frac{\bar\gamma_s}{\bar\gamma_c}
\exp\!\left[
\gamma\left(
\frac{1}{\bar\gamma_s}
-
\frac{1}{\bar\gamma_c}
\right)
\right],
\label{eq:rayleigh_tilted_ratio}
\end{equation}
where $C_p>0$ is independent of $\gamma$. Hence the MLR ordering between
$\mu_{c,p}$ and $\mu_{s,p}$ is determined solely by the average effective
SNRs $\bar\gamma_c$ and $\bar\gamma_s$.

\begin{theorem}
\label{thm:scalar_bpsk_classification}
For the scalar BPSK-CSCG/Rayleigh model, the single-mode boundary
$C_0(D)=R(G^{-1}(D))$ satisfies:
\begin{itemize}
  \item if $\bar\gamma_c>\bar\gamma_s$, then $C_0(D)$ is strictly concave on
  the interior interval, and time sharing is unnecessary;
  \item if $\bar\gamma_c<\bar\gamma_s$, then $C_0(D)$ is strictly convex on
  the interior interval, and time sharing gives a strict gain;
  \item if $\bar\gamma_c=\bar\gamma_s$, then $C_0(D)$ is linear.
\end{itemize}
\end{theorem}

\begin{proof}
When $\bar\gamma_c>\bar\gamma_s$, \eqref{eq:rayleigh_tilted_ratio} is
strictly increasing in $\gamma$, so $\mu_{c,p}$ dominates $\mu_{s,p}$ in the
MLR order. Since Lemma~\ref{lem:kappa_monotone} shows that
$\kappa_p(\gamma)$ is strictly increasing, Lemma~\ref{lem:mlr_ordering}
implies
\[
\mathbb E_{\mu_{c,p}}[\kappa_p(\Gamma)]
>
\mathbb E_{\mu_{s,p}}[\kappa_p(\Gamma)].
\]
By \eqref{eq:eta_prime_bpsk} and \eqref{eq:C0_second_eta}, this gives
$C_0''(D)<0$. The case $\bar\gamma_c<\bar\gamma_s$ is identical with the MLR
ordering reversed, giving $C_0''(D)>0$. If
$\bar\gamma_c=\bar\gamma_s$, then $w_c=w_s$, and hence
\[
R(p)+G(p)
=
\log_2\left(1+\frac{P\Omega_S}{\sigma_2^2}\right),
\qquad p\in[1/2,1],
\]
so the frontier is linear.
\end{proof}

The theorem gives a data-to-pilot interpretation of the symbol-unaware ISAC
frontier. If the communication-side effective SNR dominates, gradually
biasing the BPSK input from the equiprobable data mode toward the
deterministic pilot-like mode is already optimal. If the sensing-side
effective SNR dominates, the biased-BPSK operating points are dominated by
time sharing between the two endpoints; in the strictly convex case, the
true boundary is therefore the chord connecting the data endpoint and the
pilot-like endpoint. Fig.~\ref{fig:bpsk_scalar_cases} illustrates the three regimes.

Although the SNR ordering determines the sign of $C_0''(D)$, the kernel curvature ratio $\kappa_p(\gamma)$ is bounded for every fixed $p \in (\frac{1}{2},1)$. Hence, the time-sharing gain grows only modestly with the SNR separation.

\begin{figure}[t]
    \centering
    \includegraphics[width=\linewidth]{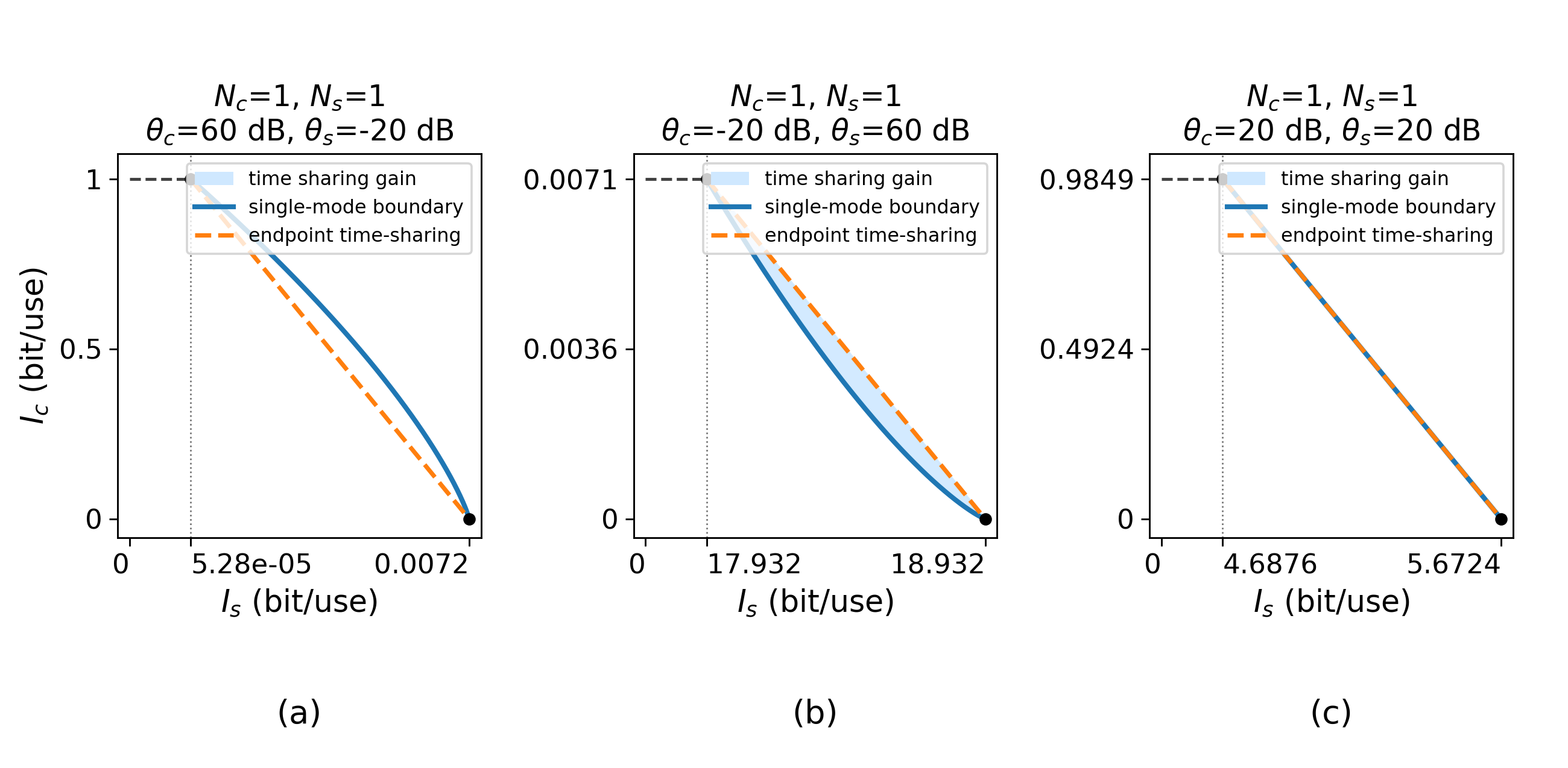}
     \caption{i.i.d.\ BPSK frontiers with a $80$ dB effective-SNR separation.}
    \label{fig:bpsk_scalar_cases}
\end{figure}

\subsection{Closed-form criterion for i.i.d.\ SIMO}

Before specializing to the i.i.d. case, we recall how the scalar curvature
criterion extends to SIMO-BPSK. Conditioned on $\bm H=\bm h$, maximal-ratio
combining reduces the communication channel to a scalar BPSK-AWGN channel
with effective SNR $\Gamma_c=2P\|\bm h\|^2/\sigma^2$. Similarly, conditioned on
$\bm S=\bm s$, the sensing-side penalty term $I(X;\bm V|\bm S=\bm s)$ is the
same scalar kernel evaluated at $\Gamma_s=2P\|\bm s\|^2/\sigma^2$. Hence
\begin{equation}
R(p)=\mathbb E[J(p,\Gamma_c)]
\end{equation}
and
\begin{equation}
G(p)=C_s-\mathbb E[J(p,\Gamma_s)].
\end{equation}
Consequently, all curvature arguments above remain valid after replacing
the exponential densities $w_c,w_s$ by the post-combining SNR densities
$f_c,f_s$. In particular, if $f_c/f_s$ is nondecreasing, the single-mode frontier $C_0(D)$ is
concave; if it is nonincreasing, the the single-mode frontier $C_0(D)$ is convex.

Consider the i.i.d.\ case. The communication covariance is $R_c=\Omega_c I_{N_c}$. The sensing covariance is $R_s=\Omega_s I_{N_s}$.

Define $\theta_c=\frac{2P\Omega_c}{\sigma^2}$ and $\theta_s=\frac{2P\Omega_s}{\sigma^2}$. Then $\Gamma_c\sim\mathrm{Gamma}(N_c,\theta_c)$, and $\Gamma_s\sim\mathrm{Gamma}(N_s,\theta_s)$.

The density ratio can be written as
\begin{equation}
\frac{f_c(\gamma)}{f_s(\gamma)}
=
C\,\gamma^{N_c-N_s}
\exp\!\left[
-\gamma\left(\frac{1}{\theta_c}-\frac{1}{\theta_s}\right)
\right],
\end{equation}
where $C>0$ is independent of $\gamma$. Taking the derivative of the logarithm yields
\begin{equation}
\frac{d}{d\gamma}\log\frac{f_c(\gamma)}{f_s(\gamma)}
=
\frac{N_c-N_s}{\gamma}
-
\left(\frac{1}{\theta_c}-\frac{1}{\theta_s}\right).
\end{equation}

Since $\frac{f_c(\gamma)}{f_s(\gamma)}>0$ on $(0,\infty)$, the monotonicity of
$\frac{f_c(\gamma)}{f_s(\gamma)}$ is equivalent to the sign of
$\frac{d}{d\gamma}\log\frac{f_c(\gamma)}{f_s(\gamma)}$ on $(0,\infty)$.

\begin{theorem}
Consider the i.i.d.\ SIMO-BPSK setting with
$R_c=\Omega_c I_{N_c}$ and $R_s=\Omega_s I_{N_s}$.
Let $\theta_c=\frac{2P\Omega_c}{\sigma^2}$ and
$\theta_s=\frac{2P\Omega_s}{\sigma^2}$.

\begin{itemize}
  \item If $N_c\ge N_s$ and $\theta_c\ge \theta_s$, then $C_0(D)$ is concave; if at least one inequality is strict, it is strictly concave, hence time-sharing is unnecessary.

  \item If $N_c\le N_s$ and $\theta_c\le \theta_s$, then $C_0(D)$ is convex; if at least one inequality is strict, it is strictly convex, hence time-sharing yields a strict gain.

  \item If $N_c=N_s$ and $\theta_c=\theta_s$, then $C_0(D)$ is linear.

  \item Otherwise (i.e., $N_c>N_s$ and $\theta_c<\theta_s$, or $N_c<N_s$ and $\theta_c>\theta_s$), the ratio $f_c(\gamma)/f_s(\gamma)$ is not monotone on $(0,\infty)$; 
\end{itemize}
\end{theorem}

In the SIMO case, the antenna numbers and per-antenna channel powers jointly determine the ordering of the post-combining SNR distributions. 
A larger number of communication antennas or a larger communication-side scale parameter shifts the communication SNR distribution toward larger values, favoring a concave frontier. 
Conversely, if the sensing side has both a larger diversity order and a larger scale parameter, the single-mode frontier becomes convex, so the corresponding single-mode region is nonconvex and time sharing with the pilot-like endpoint becomes beneficial.
When these two effects are mixed, the likelihood-ratio ordering may fail to be monotone, and the global curvature can no longer be determined by a simple dominance rule. We next validate the corollary through three i.i.d.\ SIMO-BPSK examples shown in Fig.~\ref{fig:simo_iid_cases}. 

In Fig.~\ref{fig:bpsk_scalar_cases} and Fig.~\ref{fig:simo_iid_cases}, the horizontal span of the sensing coordinate between the deterministic endpoint and the uniform-BPSK endpoint equals the nuisance-symbol information $I(X;V\mid X)$. In the high sensing-SNR regime, this term approaches $H(X)=1$ bit for uniform BPSK, reflecting the binary sign ambiguity induced by the unknown transmitted symbol. At moderate sensing SNR, the same loss term is smaller than one bit because the symbol itself is not fully resolvable from the noisy sensing observation.

\begin{figure}[t]
    \centering
    \includegraphics[width=1\columnwidth]{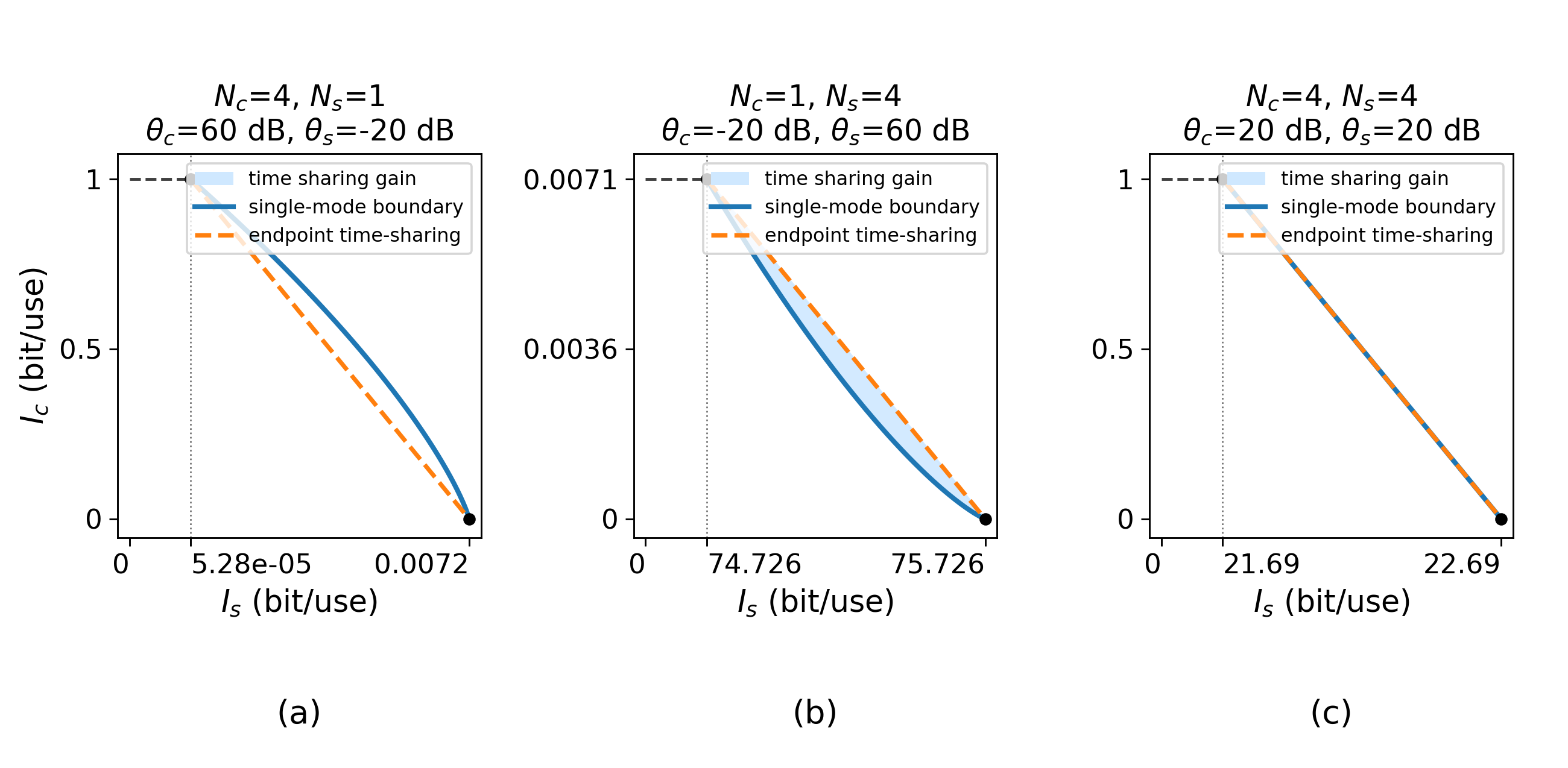}
     \caption{i.i.d.\ SIMO-BPSK frontiers with a $80$ dB effective-SNR separation.}
    \label{fig:simo_iid_cases}
\end{figure}

\section{Conclusion}
This paper studies ISAC when the sensing receiver knows the signaling law but does not have access to the realized transmitted symbols. Under this symbol-unaware sensing model, the sensing metric retains the uncertainty induced by the transmit signal, leading to a communication--sensing tradeoff fundamentally different from formulations that condition on the channel input.
For Rayleigh-fading BPSK, we derive an exact global criterion under which the single-mode frontier is strictly concave, strictly convex, or linear, depending on the ordering between the communication-side and sensing-side average effective SNRs.
These results demonstrate that, under symbol-unaware sensing, time sharing is not merely an implementation option but a structural component of the optimal capacity--sensing boundary whenever the single-mode frontier is nonconcave.

\appendices
\section{Proof of Proposition 1}\label{app:proposition 1}
For any code, the decoding error probability depends on the physical
channel only through the marginal law \(P_{Y^n|X^n}\). By memorylessness,
this marginal is
\begin{equation}
P_{Y^n|X^n}(y^n|x^n)=\prod_{i=1}^n W_c(y_i|x_i),
\end{equation}
which is the same for \(W\) and \(\widetilde W\).

Similarly, the sensing metric is symbol-by-symbol:
\begin{equation}
D_n(f_n)=\frac1n\sum_{i=1}^n I(S_i;V_i).
\end{equation}
For each coordinate \(i\), the joint law of \((S_i,V_i)\) induced by the
input marginal \(P_{X_i}\) is
\begin{equation}
P_{S_i,V_i}(s,v)
=
P_S(s)\sum_x P_{X_i}(x)W_s(v|x,s),
\end{equation}
which depends only on the sensing marginal \(W_s\). Hence the sensing
value is also the same for \(W\) and \(\widetilde W\).

Therefore every code has the same pair of reliability and sensing
performance under the two joint channels. The achievable regions are
identical.

\section{Proof of Theorem 1}
\subsection*{A. Achievability}

Fix an arbitrary pmf $P_QP_{X|Q}$ with $\mathcal Q=\{1,\ldots,K\}$. We show that every pair $(R,D)$ satisfying
$R<I(X;Y|Q)$ and $D<I(S;V|Q)$ is achievable.

Let $\epsilon>0$ and $\delta>0$. Choose positive integers $n_1,\ldots,n_K$ such that $\sum_{q=1}^K n_q=n$ and $\frac{n_q}{n}\to P_Q(q)$ for every $q$. For each $q$, choose an $n_q$-type $P_{q,n}$ such that $P_{q,n}\to P_{X|Q=q}$.

For each $q$, the constant-composition coding theorem for discrete memoryless channels guarantees a code of blocklength $n_q$ whose codewords all belong to the type class $T(P_{q,n})$, whose average error probability vanishes as $n_q\to\infty$, and whose rate approaches $I_{P_{q,n}}(X;Y)$. Here
\begin{equation}
  T(P_{q,n}) = \{ x^{n_q} \in \mathcal{X}^{n_q}: \hat{P}_{x^{n_q}} = P_{q,n} \}.
\end{equation}
Hence, for all sufficiently large $n_q$, there exists a subcode with message size $M_{q,n}$ satisfying
\begin{equation}
\bar P_{e,q}^{(n_q)}\le \frac{\epsilon}{K}
\label{eq:sub_avg_err}
\end{equation}
and
\begin{equation}
\frac{1}{n_q}\log M_{q,n}\ge I_{P_{q,n}}(X;Y)-\delta.
\label{eq:sub_rate}
\end{equation}

We concatenate these $K$ subcodes to form an overall length-$n$ code. The message set is $\mathcal W_n=\prod_{q=1}^K [1:M_{q,n}]$, and for $w=(w_1,\ldots,w_K)\in\mathcal W_n$, the encoder outputs
\begin{equation}
f_n(w)=\bigl(f_{1,n}(w_1),\ldots,f_{K,n}(w_K)\bigr).
\end{equation}
The decoder operates blockwise, applying the $q$-th subdecoder to the $q$-th received subblock.

Let $E_q=\{\hat w_q\neq w_q\}$. Since $\{\hat w\neq w\}\subseteq \cup_{q=1}^K E_q$,
\begin{equation}
P_e^{(n)}
\le
\sum_{q=1}^K \bar P_{e,q}^{(n_q)}
\le
\epsilon.
\label{eq:global_avg_err}
\end{equation}
Therefore, the overall average error probability tends to zero.

We next verify the communication rate. Since $M_n=\prod_{q=1}^K M_{q,n}$, we have
\begin{equation}
\frac{1}{n}\log M_n
=
\sum_{q=1}^K \frac{n_q}{n}\cdot \frac{1}{n_q}\log M_{q,n}.
\label{eq:rate_decomp}
\end{equation}
Combining \eqref{eq:sub_rate} and \eqref{eq:rate_decomp},
\begin{equation}
\frac{1}{n}\log M_n
\ge
\sum_{q=1}^K \frac{n_q}{n}\bigl(I_{P_{q,n}}(X;Y)-\delta\bigr).
\label{eq:rate_lower_main}
\end{equation}
Letting $n\to\infty$ and using the continuity of mutual information in the input distribution,
\begin{equation}
  \begin{aligned}
    \liminf_{n\to\infty}\frac{1}{n}\log M_n
    &\ge
    \sum_{q=1}^K P_Q(q)\,I(X;Y|Q=q)-\delta \\
    &=
    I(X;Y|Q)-\delta.
  \end{aligned}\label{eq:rate_limit_main}
\end{equation}
Since $\delta>0$ is arbitrary, every $R<I(X;Y|Q)$ is achievable.

It remains to verify the sensing requirement. Let $\mathcal I_q$ denote the coordinate set of the $q$-th subblock. For each $x\in\mathcal X$,
\begin{equation}
\frac{1}{n_q}\sum_{i\in\mathcal I_q} P_{X_i}^{(n)}(x)
=
\frac{1}{M_n}\sum_{w\in\mathcal W_n}\frac{1}{n_q}\sum_{i\in\mathcal I_q}\mathbf 1\{f_n(w)_i=x\}.
\label{eq:marginal_identity}
\end{equation}
Since every $q$-th subblock codeword lies in $T(P_{q,n})$,
\begin{equation}
\frac{1}{n_q}\sum_{i\in\mathcal I_q}\mathbf 1\{f_n(w)_i=x\}
=
P_{q,n}(x)
\label{eq:type_identity}
\end{equation}
for every $w$. Therefore,
\begin{equation}
\frac{1}{n_q}\sum_{i\in\mathcal I_q} P_{X_i}^{(n)}
=
P_{q,n}.
\label{eq:marginal_avg}
\end{equation}

\begin{lemma}
For a fixed state distribution $P_S$ and sensing channel $W_s(v|x,s)$, define
\begin{equation}
D(P_X)\triangleq I(S;V),
\end{equation}
where the joint law of $(S,V)$ induced by the input distribution $P_X$ is
\begin{equation}
P_{S,V}(s,v)=P_S(s)\sum_{x} P_X(x) W_s(v|x,s).
\end{equation}
Then $D(P_X)$ is convex in $P_X$. Namely, for any two input distributions $P_X^{(1)}$ and $P_X^{(2)}$, and any $\lambda\in[0,1]$,
\begin{equation}
D\!\left(\lambda P_X^{(1)}+(1-\lambda)P_X^{(2)}\right)
\le
\lambda D\!\left(P_X^{(1)}\right)
+(1-\lambda)D\!\left(P_X^{(2)}\right).
\end{equation}
\end{lemma}

\begin{proof}
For each fixed $s$, the conditional distribution of $V$ given $S=s$ is
\begin{equation}
P_{V|S}(v|s)=\sum_x P_X(x) W_s(v|x,s),
\end{equation}
which is affine in $P_X$. It follows that both $P_{S,V}$ and the induced marginal $P_V$ depend affinely on $P_X$. On the other hand,
\begin{equation}
I(S;V)=D\!\left(P_{S,V}\,\middle\|\,P_S P_V\right),
\end{equation}
and the relative entropy is jointly convex in its two arguments. Therefore, $I(S;V)$ is convex as a function of $P_X$, which proves the claim.
\end{proof}

We now apply the above convexity result to the empirical input distributions within each index set $\mathcal I_q$. Since $D(P_X)$ is convex in $P_X$, Jensen's inequality yields
\begin{equation}
\frac{1}{n_q}\sum_{i\in\mathcal I_q}D\bigl(P_{X_i}^{(n)}\bigr)
\ge
D\!\left(\frac{1}{n_q}\sum_{i\in\mathcal I_q}P_{X_i}^{(n)}\right)
=
D(P_{q,n}).
\label{eq:jensen_gamma}
\end{equation}
That is, the average sensing value over the positions in $\mathcal I_q$ is lower bounded by the sensing value induced by the averaged input distribution $P_{q,n}$.

Summing \eqref{eq:jensen_gamma} over all $q\in[1:K]$, and using the decomposition of the blockwise sensing metric, we obtain
\begin{equation}
D_n(f_n)
=
\frac{1}{n}\sum_{i=1}^n D\bigl(P_{X_i}^{(n)}\bigr)
\ge
\sum_{q=1}^K \frac{n_q}{n}\,D(P_{q,n}).
\label{eq:sensing_lower_main}
\end{equation}

Finally, by the convergence $n_q/n \to P_Q(q)$ and $P_{q,n}\to P_{X|Q=q}$ as $n\to\infty$, together with the continuity of mutual information in the input distribution for finite alphabets, we obtain
\begin{equation}
\liminf_{n\to\infty}D_n(f_n)
\ge
\sum_{q=1}^K P_Q(q)\,D(P_{X|Q=q})
=
I(S;V|Q).
\label{eq:sensing_limit_main}
\end{equation}
Therefore, any sensing level strictly below $I(S;V|Q)$ is achievable. In particular, every $D<I(S;V|Q)$ can be attained.

Since $P_QP_{X|Q}$ was arbitrary, the achievability of \eqref{eq:region_main} follows.

\subsection*{B. Converse}

Suppose $(R,D)$ is achievable. Then there exists a sequence of $(2^{nR_n},n)$ codes such that $R_n\to R$, $P_e^{(n)}\to 0$, and $\liminf_{n\to\infty}D_n(f_n)\ge D$. Let $W$ be uniform on $[1:2^{nR_n}]$, and let $X^n=f_n(W)$.

By Fano's inequality, there exists $\epsilon_n\to 0$ such that
\begin{equation}
R_n
\le
\frac{1}{n}I(W;Y^n)+\epsilon_n.
\label{eq:fano_rate}
\end{equation}
Since $W\to X^n\to Y^n$,
\begin{equation}
I(W;Y^n)\le I(X^n;Y^n).
\label{eq:dp_converse}
\end{equation}
Hence,
\begin{equation}
R_n
\le
\frac{1}{n}I(X^n;Y^n)+\epsilon_n.
\label{eq:block_MI_bound}
\end{equation}

By the chain rule,
\begin{equation}
I(X^n;Y^n)=\sum_{i=1}^n I(X^n;Y_i|Y^{i-1}).
\label{eq:chain_MI}
\end{equation}
Because the communication channel is memoryless,
\begin{equation}
I(X^n;Y_i|Y^{i-1})\le I(X_i;Y_i).
\label{eq:memoryless_step}
\end{equation}
Therefore,
\begin{equation}
R_n
\le
\frac{1}{n}\sum_{i=1}^n I(X_i;Y_i)+\epsilon_n.
\label{eq:sum_MI_bound}
\end{equation}

Introduce a time-sharing index $T$, uniform on $[1:n]$ and independent of everything else. Define $Q_n=T$ and $(X_n,S_n,Y_n,V_n)=(X_T,S_T,Y_T,V_T)$. Then
\begin{equation}
I(X_n;Y_n|Q_n)
=
\frac{1}{n}\sum_{i=1}^n I(X_i;Y_i).
\label{eq:single_letter_comm_main}
\end{equation}
Substituting \eqref{eq:single_letter_comm_main} into \eqref{eq:sum_MI_bound} gives
\begin{equation}
R_n\le I(X_n;Y_n|Q_n)+\epsilon_n.
\label{eq:comm_bound_single}
\end{equation}

For the sensing term, by definition of $D(\cdot)$,
\begin{equation}
I(S_i;V_i)=D\bigl(P_{X_i}^{(n)}\bigr)
\label{eq:sensing_per_letter}
\end{equation}
for each $i$. Hence,
\begin{equation}
    \begin{aligned}
        I(S_n;V_n|Q_n)
        &=
        \frac{1}{n}\sum_{i=1}^n I(S_i;V_i) \\
        &=
        \frac{1}{n}\sum_{i=1}^n D\bigl(P_{X_i}^{(n)}\bigr) \\
        &=
        D_n(f_n).
        \label{eq:sensing_single_letter_main}
    \end{aligned}
\end{equation}

Define $\mathcal A
=
\left\{
\bigl(R(P_X),D(P_X)\bigr):
P_X\in\mathcal P(\mathcal X)
\right\}
\subset \mathbb R^2.$ Since $\mathcal P(\mathcal X)$ is compact and connected, and both $R(P_X)$ and $D(P_X)$ are continuous in $P_X$, the set $\mathcal A$ is compact and connected. Moreover, for every $n$, the pair $\bigl(I(X_n;Y_n|Q_n),I(S_n;V_n|Q_n)\bigr)$ belongs to $\operatorname{conv}(\mathcal A)$.

Because $\operatorname{conv}(\mathcal A)$ is compact, there exists a subsequence along which
\begin{equation}
\bigl(I(X_n;Y_n|Q_n),I(S_n;V_n|Q_n)\bigr)
\to (\bar R,\bar D)
\in \operatorname{conv}(\mathcal A).
\end{equation}
Using \eqref{eq:comm_bound_single}, together with $R_n\to R$ and $\epsilon_n\to 0$, yields
\begin{equation}
R\le \bar R.
\label{eq:comm_converse_final}
\end{equation}
Likewise, \eqref{eq:sensing_single_letter_main} and the assumption $\liminf_{n\to\infty}D_n(f_n)\ge D$ imply
\begin{equation}
D\le \bar D.
\label{eq:sensing_converse_final}
\end{equation}

Finally, by the Fenchel--Eggleston theorem, every point in $\operatorname{conv}(\mathcal A)$ can be represented as a convex combination of at most two points in $\mathcal A$. Therefore, there exists an auxiliary random variable $Q$ with $|\mathcal Q|\le 2$ such that
\begin{align}
I(X;Y|Q)=\bar R, \\
I(S;V|Q)=\bar D.
\end{align}
Combining this representation with \eqref{eq:comm_converse_final} and \eqref{eq:sensing_converse_final}, we obtain
\begin{align}
R\le I(X;Y|Q), \\
D\le I(S;V|Q).
\end{align}
This proves the converse and completes the proof of \eqref{eq:region_main}.

\section{Derivation of the Scalar BPSK Frontier}
\label{app:scalar_bpsk_derivation}

This appendix gives the details behind the scalar BPSK expressions used in Section~IV-A.

\subsection{Communication Mutual Information}

For the scalar communication channel $Y=H X + N_1$ with $X=\sqrt P B$, $B\in\{+1,-1\}$, $H\sim\mathcal{CN}(0,\Omega_H)$, and $N_1\sim\mathcal{CN}(0,\sigma_1^2)$, conditioning on $H=h$ yields an equivalent scalar BPSK-AWGN channel.

The effective instantaneous SNR is $\gamma=\frac{2P|h|^2}{\sigma_1^2}$. The factor $2$ appears because the complex Gaussian noise is represented through an equivalent real scalar BPSK-AWGN kernel with unit real noise variance.

Let $J(p,\gamma)\triangleq I(B;\sqrt\gamma B+Z)$, where $Z\sim\mathcal N(0,1)$. Then
\begin{equation}
I(X;Y|H=h)=J\left(p,\frac{2P|h|^2}{\sigma_1^2}\right).
\end{equation}
Averaging over $H$ gives
\begin{equation}
R(p)=I(X;Y|H)
=
\mathbb E_H
\left[
J\left(p,\frac{2P|H|^2}{\sigma_1^2}\right)
\right].
\end{equation}

Since $H\sim\mathcal{CN}(0,\Omega_H)$, the random variable $|H|^2$ is exponentially distributed with density
\begin{equation}
f_{|H|^2}(t)=\frac{1}{\Omega_H}e^{-t/\Omega_H},
\qquad t\ge 0.
\end{equation}
Using the change of variables $\gamma=\frac{2Pt}{\sigma_1^2}$, $t=\frac{\sigma_1^2}{2P}\gamma$, and $dt=\frac{\sigma_1^2}{2P}d\gamma$, we obtain
\begin{align}
R(p)
&=
\int_0^\infty
J\left(p,\frac{2Pt}{\sigma_1^2}\right)
\frac{1}{\Omega_H}e^{-t/\Omega_H}dt \nonumber\\
&=
\int_0^\infty
J(p,\gamma)
\frac{\sigma_1^2}{2P\Omega_H}
e^{-\frac{\sigma_1^2}{2P\Omega_H}\gamma}
d\gamma.
\end{align}
Define $\bar\gamma_c\triangleq \frac{2P\Omega_H}{\sigma_1^2}$, and $w_c(\gamma)\triangleq
\frac{1}{\bar\gamma_c}e^{-\gamma/\bar\gamma_c}$. Then
\begin{equation}
R(p)=\int_0^\infty J(p,\gamma)w_c(\gamma)d\gamma.
\end{equation}

\subsection{Sensing Mutual Information}

For the sensing channel
\begin{equation}
V=SX+N_2,
\end{equation}
where $S\sim\mathcal{CN}(0,\Omega_S)$ and $N_2\sim\mathcal{CN}(0,\sigma_2^2)$, the chain rule gives
\begin{equation}
I(S;V)=I(S;V|X)+I(X;V)-I(X;V|S).
\label{eq:chain_sensing_app}
\end{equation}

We first show that $I(X;V)=0$. Since $S$ is zero-mean circularly symmetric complex Gaussian,
\begin{equation}
S\stackrel{d}{=}-S.
\end{equation}
Therefore,
\begin{equation}
V|X=+\sqrt P
=
\sqrt P S+N_2
\end{equation}
and
\begin{equation}
V|X=-\sqrt P
=
-\sqrt P S+N_2
\end{equation}
have the same distribution. Hence the marginal law of $V$ is independent of $X$, and
\begin{equation}
I(X;V)=0.
\end{equation}

Next, conditioned on $X=\pm\sqrt P$, the sensing observation is a scalar Gaussian channel for $S$:
\begin{equation}
V=\pm\sqrt P S+N_2.
\end{equation}
Thus
\begin{equation}
I(S;V|X)
=
\log_2\left(1+\frac{P\Omega_S}{\sigma_2^2}\right).
\end{equation}

Finally, conditioned on $S=s$, the channel from $X$ to $V$ is again an equivalent BPSK-AWGN channel with instantaneous effective SNR $\gamma=\frac{2P|s|^2}{\sigma_2^2}$. Therefore,
\begin{equation}
I(X;V|S=s)
=
J\left(p,\frac{2P|s|^2}{\sigma_2^2}\right).
\end{equation}
Averaging over $S$ yields
\begin{equation}
I(X;V|S)
=
\mathbb E_S
\left[
J\left(p,\frac{2P|S|^2}{\sigma_2^2}\right)
\right].
\end{equation}

Since $|S|^2$ is exponentially distributed with mean $\Omega_S$, the same change of variables gives
\begin{equation}
I(X;V|S)
=
\int_0^\infty J(p,\gamma)w_s(\gamma)d\gamma,
\end{equation}
where $\bar\gamma_s\triangleq\frac{2P\Omega_S}{\sigma_2^2}$, $w_s(\gamma)\triangleq
\frac{1}{\bar\gamma_s}e^{-\gamma/\bar\gamma_s}$.
Substituting these terms into \eqref{eq:chain_sensing_app}, we obtain
\begin{equation}
G(p)=I(S;V)
=
\log_2\left(1+\frac{P\Omega_S}{\sigma_2^2}\right)
-
\int_0^\infty J(p,\gamma)w_s(\gamma)d\gamma.
\end{equation}

\subsection{Curvature of the Single-mode Boundary}

Because $J(p,\gamma)$ is nonincreasing in $p$ over $p\in[1/2,1]$, the communication mutual information decreases with $p$, while the symbol-unaware sensing mutual information increases with $p$:
\begin{equation}
R'(p)\le0,
\qquad
G'(p)\ge0.
\end{equation}
Therefore the single-mode boundary can be written as
\begin{equation}
C_0(D)=R(G^{-1}(D)),
\qquad D=G(p).
\end{equation}

Differentiating with respect to $D$ gives
\begin{equation}
C_0'(D)
=
\frac{dR(p)}{dG(p)}
=
\frac{R'(p)}{G'(p)}.
\end{equation}
Define $\eta(p)\triangleq-\frac{R'(p)}{G'(p)}$. Then $C_0'(D)=-\eta(p)$. Differentiating once more with respect to $D$ gives
\begin{equation}
C_0''(D)
=
-\frac{d\eta(p)}{dD}
=
-\frac{\eta'(p)}{G'(p)}.
\end{equation}
Hence the sign of $C_0''(D)$ is opposite to the sign of $\eta'(p)$.

We now compute $\eta'(p)$. By definition,
\begin{align}
\eta'(p)
&=
-\frac{d}{dp}\left(\frac{R'(p)}{G'(p)}\right) \nonumber\\
&=
-\frac{R''(p)G'(p)-R'(p)G''(p)}{[G'(p)]^2} \nonumber\\
&=
-\frac{R'(p)}{G'(p)}
\left[
\frac{R''(p)}{R'(p)}
-
\frac{G''(p)}{G'(p)}
\right] \nonumber\\
&=
\eta(p)
\left[
\frac{R''(p)}{R'(p)}
-
\frac{G''(p)}{G'(p)}
\right].
\label{eq:eta_prime_ratio_app}
\end{align}

Next, observe that
\begin{equation}
R'(p)=\int_0^\infty J_p(p,\gamma)w_c(\gamma)d\gamma,
\end{equation}
and
\begin{equation}
R''(p)=\int_0^\infty J_{pp}(p,\gamma)w_c(\gamma)d\gamma.
\end{equation}
Similarly,
\begin{equation}
G'(p)=-\int_0^\infty J_p(p,\gamma)w_s(\gamma)d\gamma,
\end{equation}
and
\begin{equation}
G''(p)=-\int_0^\infty J_{pp}(p,\gamma)w_s(\gamma)d\gamma.
\end{equation}

For $p\in(1/2,1)$, define $a_p(\gamma)\triangleq -J_p(p,\gamma)>0$ and $\kappa_p(\gamma)\triangleq
\frac{J_{pp}(p,\gamma)}{J_p(p,\gamma)}.$ Also define
$
A_c(p)\triangleq
\int_0^\infty a_p(\gamma)w_c(\gamma)d\gamma,
\quad
A_s(p)\triangleq
\int_0^\infty a_p(\gamma)w_s(\gamma)d\gamma.
$
Then
\begin{equation}
\mu_{c,p}(\gamma)
=
\frac{a_p(\gamma)w_c(\gamma)}{A_c(p)},
\end{equation}
and
\begin{equation}
\mu_{s,p}(\gamma)
=
\frac{a_p(\gamma)w_s(\gamma)}{A_s(p)}
\end{equation}
are probability densities on $(0,\infty)$.

Using these definitions,
\begin{align}
\frac{R''(p)}{R'(p)}
&=
\frac{
\int_0^\infty J_{pp}(p,\gamma)w_c(\gamma)d\gamma
}{
\int_0^\infty J_p(p,\gamma)w_c(\gamma)d\gamma
} \nonumber\\
&=
\int_0^\infty
\frac{J_{pp}(p,\gamma)}{J_p(p,\gamma)}
\frac{-J_p(p,\gamma)w_c(\gamma)}
{\int_0^\infty -J_p(p,t)w_c(t)dt}
d\gamma \nonumber\\
&=
\int_0^\infty
\kappa_p(\gamma)\mu_{c,p}(\gamma)d\gamma \nonumber\\
&=
\mathbb E_{\mu_{c,p}}[\kappa_p(\Gamma)].
\end{align}
Likewise,
\begin{equation}
\frac{G''(p)}{G'(p)}
=
\mathbb E_{\mu_{s,p}}[\kappa_p(\Gamma)].
\end{equation}
Substituting into \eqref{eq:eta_prime_ratio_app} yields
\begin{equation}
\eta'(p)
=
\eta(p)
\left(
\mathbb E_{\mu_{c,p}}[\kappa_p(\Gamma)]
-
\mathbb E_{\mu_{s,p}}[\kappa_p(\Gamma)]
\right).
\end{equation}

\subsection{Likelihood-Ratio Reduction in the Rayleigh Case}

The tilted densities satisfy
\begin{equation}
\frac{\mu_{c,p}(\gamma)}{\mu_{s,p}(\gamma)}
=
\frac{A_s(p)}{A_c(p)}
\frac{w_c(\gamma)}{w_s(\gamma)}.
\end{equation}
Since $A_s(p)/A_c(p)>0$ is independent of $\gamma$, the monotonicity of $\mu_{c,p}/\mu_{s,p}$ is the same as that of $w_c/w_s$.

For the Rayleigh fading case,
\begin{align}
\frac{w_c(\gamma)}{w_s(\gamma)}
&=
\frac{
\frac{1}{\bar\gamma_c}e^{-\gamma/\bar\gamma_c}
}{
\frac{1}{\bar\gamma_s}e^{-\gamma/\bar\gamma_s}
} \nonumber\\
&=
\frac{\bar\gamma_s}{\bar\gamma_c}
\exp\left[
\gamma
\left(
\frac{1}{\bar\gamma_s}
-
\frac{1}{\bar\gamma_c}
\right)
\right].
\end{align}
Therefore,
\begin{itemize}
\item if $\bar\gamma_c>\bar\gamma_s$, then $w_c(\gamma)/w_s(\gamma)$ is strictly increasing in $\gamma$;
\item if $\bar\gamma_c<\bar\gamma_s$, then $w_c(\gamma)/w_s(\gamma)$ is strictly decreasing in $\gamma$;
\item if $\bar\gamma_c=\bar\gamma_s$, then $w_c=w_s$.
\end{itemize}
Combining this likelihood-ratio ordering with the monotonicity of $\kappa_p(\gamma)$ proves the scalar BPSK classification theorem in Section~IV-B.

\section{Proof of Lemma 2}
\label{app:proof_kappa_monotone}

For convenience, all logarithms in this appendix are natural logarithms. 

Set $m=2p-1\in(0,1)$, so that $p=(1+m)/2$. Writing the same kernel as $J(m,\gamma)$, the chain rule gives $J_p=2J_m$ and $J_{pp}=4J_{mm}$. Hence
\begin{equation}\label{eq:kappa-reduction}
\kappa_p(\gamma)=2\,\frac{-J_{mm}(m,\gamma)}{-J_m(m,\gamma)}.
\end{equation}
Thus it suffices to prove that
\begin{equation}\label{eq:goal-ratio}
\gamma\mapsto \frac{-J_{mm}(m,\gamma)}{-J_m(m,\gamma)}
\end{equation}
is strictly increasing.

The proof has three steps. First, we rewrite the derivatives of $J$ as expectations of two scalar functions of a random variable $U_\gamma\in(0,1)$. Second, we show that the family $\{U_\gamma\}_{\gamma>0}$ is strictly increasing in the monotone likelihood-ratio (MLR) order. Third, we verify that the relevant likelihood-ratio functional is strictly increasing in $U_\gamma$, which implies the desired monotonicity.

\subsection{Derivative representation}

Consider \eqref{eq:com_channel}, let
\begin{equation}\label{eq:c-def}
c_\gamma(y)\triangleq \phi(y)e^{-\gamma/2}\cosh(\sqrt{\gamma}\,y),
\end{equation}
\begin{equation}\label{eq:t-def}
t_\gamma(y)\triangleq \tanh(\sqrt{\gamma}\,y),
\end{equation}
where $\phi(y)=(2\pi)^{-1/2}e^{-y^2/2}$. A direct calculation shows that the output density can be factorized as
\begin{equation}\label{eq:q-factorization}
q_{m,\gamma}(y)=c_\gamma(y)\bigl(1+m t_\gamma(y)\bigr).
\end{equation}
Moreover, $c_\gamma$ is even, $t_\gamma$ is odd, and $c_\gamma=q_{0,\gamma}$ is a probability density.

Since $J(m,\gamma)=h(Y)-h(Y|B)$ and $h(Y|B)=\tfrac12\log(2\pi e)$, we have
\begin{equation}\label{eq:J-entropy}
J(m,\gamma)=-\int_{\mathbb R} q_{m,\gamma}(y)\log q_{m,\gamma}(y)\,dy-\frac12\log(2\pi e).
\end{equation}
Substituting \eqref{eq:q-factorization} into \eqref{eq:J-entropy} and using the parity of $c_\gamma$, $t_\gamma$, and $\log c_\gamma$, we obtain
\begin{equation}\label{eq:J-reduced}
J(m,\gamma)=C_\gamma-\int_{\mathbb R} c_\gamma(y)\bigl(1+m t_\gamma(y)\bigr)
\log\bigl(1+m t_\gamma(y)\bigr)\,dy,
\end{equation}
where $C_\gamma$ depends only on $\gamma$.

Now fix an arbitrary compact interval $K=[m_0,m_1]\subset(0,1)$. Since $t_\gamma(y)\in(-1,1)$ for all $y$, one has $1+m t_\gamma(y)\in[1-m_1,1+m_1]$ for every $m\in K$. Therefore,
\begin{align}
\left|\frac{\partial}{\partial m}\Big[(1+m t_\gamma(y))\log(1+m t_\gamma(y))\Big]\right|
&\le C_{1,K}, \label{eq:dom-1}\\
\left|\frac{\partial^2}{\partial m^2}\Big[(1+m t_\gamma(y))\log(1+m t_\gamma(y))\Big]\right|
&\le C_{2,K}, \label{eq:dom-2}
\end{align}
for finite constants $C_{1,K},C_{2,K}$ independent of $y$. Because $c_\gamma$ is a probability density, the bounds in \eqref{eq:dom-1}--\eqref{eq:dom-2} are integrable against $c_\gamma(y)\,dy$. Hence the dominated convergence theorem allows differentiation under the integral sign in \eqref{eq:J-reduced}, yielding
\begin{align}
J_m(m,\gamma)
&=-\int_{\mathbb R} c_\gamma(y)
 t_\gamma(y)\log\bigl(1+m t_\gamma(y)\bigr)\,dy, \label{eq:Jm}\\
J_{mm}(m,\gamma)
&=-\int_{\mathbb R} c_\gamma(y)
\frac{t_\gamma(y)^2}{1+m t_\gamma(y)}\,dy. \label{eq:Jmm}
\end{align}

Let $Y_0\sim c_\gamma$, and define $T_\gamma\triangleq t_\gamma(Y_0)$ and $U_\gamma\triangleq |T_\gamma|$. Because $c_\gamma$ is even and $t_\gamma$ is odd, $T_\gamma$ is symmetric about zero. Symmetrizing \eqref{eq:Jm} and \eqref{eq:Jmm} gives
\begin{align}
-J_m(m,\gamma)
&=\E\!\left[U_\gamma\,\artanh(mU_\gamma)\right], \label{eq:Jm-sym}\\
-J_{mm}(m,\gamma)
&=\E\!\left[\frac{U_\gamma^2}{1-m^2U_\gamma^2}\right]. \label{eq:Jmm-sym}
\end{align}
Define
\begin{align}
A_m(u)&\triangleq \frac{u^2}{1-m^2u^2}, \label{eq:Am-def}\\
B_m(u)&\triangleq u\,\artanh(mu), \qquad u\in(0,1). \label{eq:Bm-def}
\end{align}
Then \eqref{eq:kappa-reduction}, \eqref{eq:Jm-sym}, and \eqref{eq:Jmm-sym} imply
\begin{equation}\label{eq:kappa-AB}
\kappa_p(\gamma)=2\,\frac{\E[A_m(U_\gamma)]}{\E[B_m(U_\gamma)]}.
\end{equation}

\subsection{MLR monotonicity of the transformed variable}

Since $c_\gamma=q_{0,\gamma}$, the random variable $Y_0$ has the same law as the channel output under equiprobable input. Hence
\begin{equation}\label{eq:Y0-law}
Y_0\overset{d}=\sqrt{\gamma}\,B+Z,
\qquad \Pr(B=\pm1)=\frac12.
\end{equation}
Multiplying by $\sqrt{\gamma}$ and using the symmetry of $B$ and $Z$ yields
\begin{equation}\label{eq:Ug-Vg}
U_\gamma=\bigl|\tanh(\sqrt{\gamma}\,Y_0)\bigr|
\overset{d}=\tanh\bigl(|\gamma+\sqrt{\gamma}Z|\bigr).
\end{equation}
Define $V_\gamma\triangleq |\gamma+\sqrt{\gamma}Z|$, so that $U_\gamma=\tanh(V_\gamma)$.

For $v\ge0$, the density of $V_\gamma$ is
\begin{align}
f_{V_\gamma}(v)
&=\frac{1}{\sqrt{2\pi\gamma}}
\left(e^{-(v-\gamma)^2/(2\gamma)}+e^{-(v+\gamma)^2/(2\gamma)}\right) \nonumber\\
&=\frac{2\cosh v}{\sqrt{2\pi\gamma}}
e^{-\frac{\gamma}{2}-\frac{v^2}{2\gamma}}. \label{eq:V-density}
\end{align}
Therefore, for any $\gamma_2>\gamma_1>0$,
\begin{equation}\label{eq:V-LR}
\frac{f_{V_{\gamma_2}}(v)}{f_{V_{\gamma_1}}(v)}
=\sqrt{\frac{\gamma_1}{\gamma_2}}
e^{\left(-\frac{\gamma_2-\gamma_1}{2}\right)}
e^{\left[\frac{v^2}{2}\left(\frac1{\gamma_1}-\frac1{\gamma_2}\right)\right]}.
\end{equation}
Because $\frac1{\gamma_1}-\frac1{\gamma_2}>0$, the right-hand side of \eqref{eq:V-LR} is strictly increasing in $v$. Thus the family $\{V_\gamma\}_{\gamma>0}$ is strictly increasing in the MLR order.

Next, note that $u=\tanh v$ is a strictly increasing $C^1$ bijection from $(0,\infty)$ onto $(0,1)$, with inverse $v=\artanh u$. By change of variables,
\begin{equation}\label{eq:U-density}
f_{U_\gamma}(u)=\frac{f_{V_\gamma}(\artanh u)}{1-u^2},
\qquad 0<u<1.
\end{equation}
Hence, for $\gamma_2>\gamma_1$,
\begin{equation}\label{eq:U-LR}
\frac{f_{U_{\gamma_2}}(u)}{f_{U_{\gamma_1}}(u)}
=\frac{f_{V_{\gamma_2}}(\artanh u)}{f_{V_{\gamma_1}}(\artanh u)},
\end{equation}
which is strictly increasing in $u$ because both the likelihood ratio in \eqref{eq:V-LR} and the map $u\mapsto\artanh u$ are strictly increasing. Therefore, $\{U_\gamma\}_{\gamma>0}$ is also strictly increasing in the MLR order.

\subsection{Monotonicity of the ratio functional}

For $u\in(0,1)$, consider the ratio
\begin{equation}\label{eq:psi-def}
\psi(u)\triangleq \frac{A_m(u)}{B_m(u)}
=\frac{u}{(1-m^2u^2)\artanh(mu)}.
\end{equation}
Let $x=mu\in(0,1)$ and define
\begin{equation}\label{eq:h-def}
h(x)\triangleq \frac{x}{(1-x^2)\artanh x}.
\end{equation}
Then $\psi(u)=m^{-1}h(mu)$. Hence it is enough to prove that $h$ is strictly increasing on $(0,1)$. Differentiating $\log h(x)$ gives
\begin{equation}\label{eq:h-log-deriv}
\frac{h'(x)}{h(x)}=
\frac{(1+x^2)\artanh x-x}{x(1-x^2)\artanh x}.
\end{equation}
Let
\begin{equation}\label{eq:g-def}
g(x)\triangleq (1+x^2)\artanh x-x.
\end{equation}
Since $g(0)=0$ and
\begin{equation}\label{eq:g-prime}
g'(x)=2x\artanh x+\frac{2x^2}{1-x^2}>0,
\qquad x\in(0,1),
\end{equation}
we have $g(x)>0$ for all $x\in(0,1)$. By \eqref{eq:h-log-deriv}, this implies $h'(x)>0$ on $(0,1)$, and therefore $\psi(u)$ is strictly increasing on $(0,1)$.

Now define the $B_m$-tilted density
\begin{equation}\label{eq:tilted-density}
\tilde f_\gamma(u)
\triangleq
\frac{B_m(u)f_{U_\gamma}(u)}{\int_0^1 B_m(t)f_{U_\gamma}(t)\,dt},
\qquad 0<u<1.
\end{equation}
Because $B_m(u)>0$ on $(0,1)$, for any $\gamma_2>\gamma_1$ we have
\begin{equation}\label{eq:tilted-LR}
\frac{\tilde f_{\gamma_2}(u)}{\tilde f_{\gamma_1}(u)}
=C_{\gamma_1,\gamma_2}
\frac{f_{U_{\gamma_2}}(u)}{f_{U_{\gamma_1}}(u)},
\end{equation}
where $C_{\gamma_1,\gamma_2}>0$ does not depend on $u$. Thus the family $\{\tilde f_\gamma\}_{\gamma>0}$ is also strictly increasing in the MLR order.

Fix $\gamma_2>\gamma_1$ and write
\begin{equation}\label{eq:r-def}
r(u)\triangleq \frac{\tilde f_{\gamma_2}(u)}{\tilde f_{\gamma_1}(u)}.
\end{equation}
Since $r$ is strictly increasing and $\int_0^1 r(u)\tilde f_{\gamma_1}(u)\,du=1$, there exists a unique $u_\star\in(0,1)$ such that
\begin{equation}\label{eq:single-crossing}
\tilde f_{\gamma_2}(u)-\tilde f_{\gamma_1}(u)
\begin{cases}
<0, & 0<u<u_\star,\\
=0, & u=u_\star,\\
>0, & u_\star<u<1.
\end{cases}
\end{equation}
Since $\psi$ is strictly increasing, the product
\begin{equation}\label{eq:positive-product}
\bigl(\psi(u)-\psi(u_\star)\bigr)
\bigl(\tilde f_{\gamma_2}(u)-\tilde f_{\gamma_1}(u)\bigr)
\end{equation}
is nonnegative for all $u\in(0,1)$ and strictly positive except at $u=u_\star$. Therefore,
\begin{align}
\Delta_\psi(\gamma_2,\gamma_1)
&\triangleq \E_{\tilde f_{\gamma_2}}[\psi(U)]-\E_{\tilde f_{\gamma_1}}[\psi(U)] \nonumber\\
&=\int_0^1 \psi(u)\bigl(\tilde f_{\gamma_2}(u)-\tilde f_{\gamma_1}(u)\bigr)\,du \nonumber\\
&=\int_0^1 \bigl(\psi(u)-\psi(u_\star)\bigr)
\bigl(\tilde f_{\gamma_2}(u)-\tilde f_{\gamma_1}(u)\bigr)\,du \nonumber\\
&>0. \label{eq:strict-expectation}
\end{align}
Hence $\gamma\mapsto \E_{\tilde f_\gamma}[\psi(U)]$ is strictly increasing on $(0,\infty)$.

Finally, by the definition of the tilted density in \eqref{eq:tilted-density},
\begin{equation}\label{eq:tilted-identity}
\E_{\tilde f_\gamma}\!\left[\frac{A_m(U)}{B_m(U)}\right]
=\frac{\E[A_m(U_\gamma)]}{\E[B_m(U_\gamma)]}.
\end{equation}
Combining \eqref{eq:kappa-AB}, \eqref{eq:strict-expectation}, and \eqref{eq:tilted-identity} yields that
\begin{equation}\label{eq:ratio-final}
\gamma\mapsto
\frac{\E[A_m(U_\gamma)]}{\E[B_m(U_\gamma)]}
=\frac{-J_{mm}(m,\gamma)}{-J_m(m,\gamma)}
\end{equation}
is strictly increasing. By \eqref{eq:kappa-reduction}, $\gamma\mapsto \kappa_p(\gamma)$ is strictly increasing on $(0,\infty)$. This completes the proof.

\section{Proof of Lemma 3}
\label{app:proof_mlr_ordering}

Define
\begin{equation}
\Delta \triangleq \mathbb E_{q_1}[\phi(X)]-\mathbb E_{q_2}[\phi(X)].
\end{equation}
Since $q_1=rq_2$, we have
\begin{equation}
\Delta
=
\int_{\mathcal X}\phi(x)\bigl(r(x)-1\bigr)q_2(x)\,dx.
\end{equation}
Moreover, because both $q_1$ and $q_2$ are probability densities,
\begin{equation}
\int_{\mathcal X} r(x)q_2(x)\,dx
=
\int_{\mathcal X} q_1(x)\,dx
=
1.
\end{equation}
Hence,
\begin{equation}
\Delta
=
\int_{\mathcal X}\phi(x)r(x)q_2(x)\,dx
-
\int_{\mathcal X}\phi(x)q_2(x)\,dx
\int_{\mathcal X}r(x)q_2(x)\,dx.
\end{equation}
If $X\sim q_2$, then
\begin{equation}
\Delta
=
\operatorname{Cov}_{q_2}\bigl(\phi(X),r(X)\bigr).
\end{equation}

Now let $X$ and $Y$ be i.i.d.\ random variables with density $q_2$. Using the symmetric representation of covariance, we obtain
\begin{equation}
2\,\operatorname{Cov}_{q_2}\bigl(\phi(X),r(X)\bigr)
=
\mathbb E\!\left[(\phi(X)-\phi(Y))(r(X)-r(Y))\right].
\end{equation}
Therefore,
\begin{equation}
\Delta
=
\frac12\,
\mathbb E\!\left[(\phi(X)-\phi(Y))(r(X)-r(Y))\right].
\end{equation}
Equivalently,
\begin{equation}
\Delta
=
\frac12
\iint_{\mathcal X\times\mathcal X}
(\phi(x)-\phi(y))(r(x)-r(y))
\,q_2(x)q_2(y)\,dx\,dy.
\end{equation}

Since both $\phi$ and $r$ are nondecreasing, for every $x,y\in\mathcal X$,
\begin{equation}
(\phi(x)-\phi(y))(r(x)-r(y))\ge 0.
\end{equation}
Thus the integrand is everywhere nonnegative, which implies $\Delta\ge 0$. Hence,
\begin{equation}
\mathbb E_{q_1}[\phi(X)]\ge \mathbb E_{q_2}[\phi(X)].
\end{equation}

If, in addition, both $\phi$ and $r$ are strictly increasing, then for every $x\neq y$,
\begin{equation}
(\phi(x)-\phi(y))(r(x)-r(y))>0.
\end{equation}
Since $X$ and $Y$ both admit the density $q_2$, we have
\begin{equation}
\mathbb P(X=Y)=0.
\end{equation}
Therefore,
\begin{equation}
(\phi(X)-\phi(Y))(r(X)-r(Y))>0
\qquad \text{almost surely},
\end{equation}
and hence
\begin{equation}
\Delta
=
\frac12\,
\mathbb E\!\left[(\phi(X)-\phi(Y))(r(X)-r(Y))\right]
>0.
\end{equation}
It follows that
\begin{equation}
\mathbb E_{q_1}[\phi(X)]>\mathbb E_{q_2}[\phi(X)].
\end{equation}

\section{Numerical Evaluation of Frontier Curvature and Boundedness}
\label{app:numerical_curvature}

\begin{table*}[!t]
\centering
\caption{Normalized curvature and endpoint chord deviation for scalar BPSK and i.i.d. SIMO-BPSK examples.}
\label{tab:curvature_numerics}
\begin{tabular}{c c c c c c c}
\hline
Model & \((N_c,N_s)\) & SNR pair [dB] & Shape 
& \(\widetilde{\mathcal K}(0.8)\)
& \(\max_{p\in[0.55,0.95]}|\widetilde{\mathcal K}(p)|\)
& \(\Delta_{\rm TS,n}\) / \(\Delta_{\rm ch,n}\) \\
\hline
Scalar BPSK & -- & \((20,-20)\) & Concave
& \(-0.177\) & \(0.426\) & \(0.000/0.106\) \\
Scalar BPSK & -- & \((-20,20)\) & Convex
& \(+0.177\) & \(0.426\) & \(0.106/0.106\) \\
Scalar BPSK & -- & \((-30,50)\) & Convex
& \(+0.184\) & \(0.436\) & \(0.109/0.109\) \\
Scalar BPSK & -- & \((-20,60)\) & Convex
& \(+0.178\) & \(0.429\) & \(0.107/0.107\) \\
Scalar BPSK & -- & \((0,0)\) & Linear
& \(0.000\) & \(0.000\) & \(0.000/0.000\) \\
\hline
SIMO-BPSK & \((4,1)\) & \((0,0)\) & Concave
& \(-0.042\) & \(0.117\) & \(0.000/0.031\) \\
SIMO-BPSK & \((1,4)\) & \((0,0)\) & Convex
& \(+0.042\) & \(0.117\) & \(0.031/0.031\) \\
SIMO-BPSK & \((1,4)\) & \((-20,20)\) & Convex
& \(+0.178\) & \(0.429\) & \(0.107/0.107\) \\
SIMO-BPSK & \((2,2)\) & \((0,0)\) & Linear
& \(0.000\) & \(0.000\) & \(0.000/0.000\) \\
\hline
\end{tabular}
\end{table*}

This appendix describes how the curvature values reported in the numerical examples are computed. 
The discussion also clarifies why a large SNR separation does not necessarily produce a proportionally large boundary deformation.

Recall that the single-mode frontier is parameterized as \eqref{eq:C0_def_bpsk}-\eqref{eq:C0_second_eta}.
Thus, the sign of \(C_0''(D)\) determines whether the raw frontier is locally concave or convex.

To compare curvature values across different SNR settings, we use normalized coordinates
\begin{equation}
    x(p)
    =
    \frac{G(p)-G(1/2)}{G(1)-G(1/2)}
    =
    1-\frac{A_s(p)}{A_s(1/2)},
    \label{eq:app_normalized_D}
\end{equation}
and
\begin{equation}
    y(p)
    =
    \frac{R(p)}{R(1/2)}
    =
    \frac{A_c(p)}{A_c(1/2)}.
    \label{eq:app_normalized_R}
\end{equation}
The normalized signed geometric curvature is then computed as
\begin{equation}
    \widetilde{\mathcal K}(p)
    =
    \frac{x'(p)y''(p)-y'(p)x''(p)}
    {\left([x'(p)]^2+[y'(p)]^2\right)^{3/2}}.
    \label{eq:app_normalized_curvature}
\end{equation}
Here
\begin{align}
x'(p)  &=-\frac{A_s'(p)}{A_s\!\left(\frac12\right)}, \\
x''(p) &=-\frac{A_s''(p)}{A_s\!\left(\frac12\right)}, \\
y'(p)  &= \frac{A_c'(p)}{A_c\!\left(\frac12\right)}, \\
y''(p) &= \frac{A_c''(p)}{A_c\!\left(\frac12\right)} .
\end{align}
The derivatives are obtained by differentiating under the expectation:
\begin{align}
    A_i'(p) &= \mathbb E[J_p(p,\Gamma_i)], \\
    A_i''(p) &= \mathbb E[J_{pp}(p,\Gamma_i)].
\end{align}

For numerical integration, the Gaussian expectation inside \(J(p,\gamma)\), \(J_p(p,\gamma)\), and \(J_{pp}(p,\gamma)\) is evaluated using Gauss-Hermite quadrature. The outer expectation over \(\Gamma_i\) is evaluated using Gauss-Laguerre quadrature in the scalar Rayleigh case and generalized Gauss-Laguerre quadrature in the i.i.d. SIMO case. Specifically, for the scalar exponential case,
\begin{equation}
    \begin{aligned}
        A_i^{(r)}(p)
        &=
        \mathbb E\!\left[
        \frac{\partial^r}{\partial p^r}J(p,\Gamma_i)
        \right] \\
        &\approx
        \sum_{\ell=1}^{L}
        w_\ell
        \frac{\partial^r}{\partial p^r}
        J(p,\bar\gamma_i t_\ell),
        \qquad r=0,1,2,
        \label{eq:app_laguerre_exp}
    \end{aligned}
\end{equation}

where \(\{t_\ell,w_\ell\}\) are the standard Gauss-Laguerre nodes and weights. For the i.i.d. SIMO case,
\begin{equation}
    A_i^{(r)}(p)
    \approx
    \frac{1}{\Gamma(N_i)}
    \sum_{\ell=1}^{L}
    w_\ell^{(N_i-1)}
    \frac{\partial^r}{\partial p^r}
    J(p,\theta_i t_\ell),
    \qquad r=0,1,2,
    \label{eq:app_laguerre_gamma}
\end{equation}
where \(\{t_\ell,w_\ell^{(N_i-1)}\}\) are the generalized Gauss-Laguerre nodes and weights associated with the weight \(t^{N_i-1}e^{-t}\).

In addition to local curvature, we also report the normalized endpoint chord deviation. In the normalized coordinates \((x,y)\), the endpoint time-sharing line is simply
\begin{equation}
    y_{\rm TS}(x)=1-x.
\end{equation}
Therefore, the normalized endpoint time-sharing gain is
\begin{equation}
    \Delta_{\rm TS,n}
    =
    \max_{p\in[1/2,1]}
    \left[1-x(p)-y(p)\right]_+,
    \label{eq:app_normalized_ts_gain}
\end{equation}
where \([a]_+=\max\{a,0\}\). We also define the normalized absolute chord deviation
\begin{equation}
    \Delta_{\rm ch,n}
    =
    \max_{p\in[1/2,1]}
    \left|1-x(p)-y(p)\right|.
    \label{eq:app_normalized_chord_deviation}
\end{equation}
The quantity \(\Delta_{\rm TS,n}\) measures the actual gain from endpoint time sharing, whereas \(\Delta_{\rm ch,n}\) measures the global deviation from the endpoint chord irrespective of whether the frontier is concave or convex.

Table~\ref{tab:curvature_numerics} gives representative numerical values. The curvature maximum is computed over \(p\in[0.55,0.95]\), excluding the two degenerate endpoints where both parametric derivatives may become numerically ill-conditioned. The scalar rows use the SNR pair \((\bar\gamma_c,\bar\gamma_s)\) in dB, while the SIMO rows use the per-branch SNR pair \((\theta_c,\theta_s)\) in dB. The values were obtained with \(80\) Gauss-Hermite nodes and \(100\) Gauss-Laguerre or generalized Gauss-Laguerre nodes.

The table shows that the sign of the reported signed curvature follows the analytical concavity criterion, but the normalized curvature and the normalized chord deviation remain moderate. This explains why, even when time sharing is theoretically necessary in the convex cases, the practical gain over a properly selected single-mode input can be limited.

This behavior can be partially understood from the boundedness of the kernel curvature ratio \eqref{eq:kappa_def}. For every fixed \(p\in(1/2,1)\), \(\kappa_p(\gamma)\) is strictly increasing in \(\gamma\), but it is bounded between its low-SNR and high-SNR limits:
\begin{align}
    \lim_{\gamma\rightarrow 0}\kappa_p(\gamma)
    &=
    \frac{2}{2p-1}, \\
    \lim_{\gamma\to\infty}\kappa_p(\gamma)
    &=
    \frac{1}{p(1-p)\ln\frac{p}{1-p}}.
    \label{eq:app_kappa_limits}
\end{align}
Hence, for $\gamma > 0$,
\begin{equation}
    \frac{2}{2p-1}
    \le
    \kappa_p(\gamma)
    \le
    \frac{1}{p(1-p)\ln\frac{p}{1-p}}.
    \label{eq:app_kappa_bounded}
\end{equation}
Consequently, the SNR ordering can determine the sign of the frontier curvature through the stochastic ordering of the effective SNR distributions, but increasing the SNR separation cannot make the kernel contribution to the curvature grow without bound. This boundedness should be interpreted for each fixed \(p\), or uniformly over compact subintervals of \((1/2,1)\); it is not a uniform statement over the entire open interval because the limiting expressions in \eqref{eq:app_kappa_limits} become singular near the endpoints.


\begin{thebibliography}{99}

\bibitem{Liu2022JSAC}
F. Liu \emph{et al.}, “Integrated sensing and communications: Toward dual-functional wireless networks for 6G and beyond,” \emph{IEEE J. Sel. Areas Commun.}, vol. 40, no. 6, pp. 1728--1767, Jun. 2022.

\bibitem{Ahmadipour2024ISAC}
M. Ahmadipour, M. Kobayashi, M. Wigger, and G. Caire, “An information-theoretic approach to joint sensing and communication,” \emph{IEEE Trans. Inf. Theory}, vol. 70, no. 2, pp. 1124--1146, Feb. 2024.

\bibitem{Wei2026JSAC}
Z. Wei, J. Piao, L. Wang, H. Wu, and Z. Feng, “Mutual information of MIMO-OFDM integrated sensing and communication system in space--time--frequency domains,” \emph{IEEE J. Sel. Areas Commun.}, vol. 44, pp. 92--106, 2026.

\bibitem{Li2024TVT}
J. Li, G. Zhou, T. Gong, and N. Liu, “A framework for mutual information-based MIMO integrated sensing and communication beamforming design,” \emph{IEEE Trans. Veh. Technol.}, vol. 73, no. 6, pp. 8352--8366, Jun. 2024.

\bibitem{Chen2026TWC}
G. Chen \emph{et al.}, “Movable antenna-enabled MIMO integrated sensing and communication: A unified mutual information framework,” \emph{IEEE Trans. Wireless Commun.}, vol. 25, pp. 14440--14454, 2026.

\bibitem{Xiong2023GaussianISAC}
Y. Xiong, F. Liu, Y. Cui, W. Yuan, T. X. Han, and G. Caire, “On the fundamental tradeoff of integrated sensing and communications under Gaussian channels,” \emph{IEEE Trans. Inf. Theory}, vol. 69, no. 9, pp. 5723--5751, Sep. 2023.

\bibitem{Liu2023DeterministicRandom}
F. Liu, Y. Xiong, K. Wan, T. X. Han, and G. Caire, “Deterministic-random tradeoff of integrated sensing and communications in Gaussian channels: A rate-distortion perspective,” in \emph{Proc. IEEE Int. Symp. Inf. Theory (ISIT)}, Taipei, Taiwan, 2023, pp. xxx--xxx.

\bibitem{Jiao2025BistaticISAC}
T. Jiao \emph{et al.}, “Information-theoretic limits of bistatic integrated sensing and communication,” \emph{IEEE Trans. Inf. Theory}, vol. 71, no. 12, pp. 9302--9318, Dec. 2025.

\bibitem{Xie2024RandomSMI}
L. Xie, F. Liu, J. Luo, and S. Song, “Sensing mutual information with random signals in Gaussian channels,” \emph{IEEE Trans. Commun.}, vol. 73, no. 10, pp. 9437--9452, Oct. 2025.

\bibitem{YLiu2026TIT}
Y. Liu, M. Li, L. Ong, and A. Yener, “Fundamental limits of bistatic integrated sensing and communications over memoryless relay channels,” \emph{IEEE Trans. Inf. Theory}, early access, 2026.

\bibitem{JChen2026TCOM}
J. Chen, L. Yu, Y. Li, W. Shi, Y. Ge, and W. Tong, “On the fundamental limits of integrated sensing and communications under logarithmic loss,” \emph{IEEE Trans. Commun.}, vol. 74, pp. 46--60, 2026.

\bibitem{SLu2026TN}
S. Lu \emph{et al.}, “Sensing with random communication signals,” \emph{IEEE Netw.}, vol. 40, no. 1, pp. 98--106, Jan. 2026.

\bibitem{Fliu2026JSAC}
F. Liu \emph{et al.}, “Sensing with communication signals: From information theory to signal processing,” \emph{IEEE J. Sel. Areas Commun.}, vol. 44, pp. 1--30, 2026.

\bibitem{WZhang2011}
W. Zhang, S. Vedantam, and U. Mitra, “Joint transmission and state estimation: A constrained channel coding approach,” \emph{IEEE Trans. Inf. Theory}, vol. 57, no. 10, pp. 7084--7095, Oct. 2011.

\end{thebibliography}
\end{document}